\documentclass{ifacconf}
\usepackage[T1]{fontenc}

\usepackage{graphicx}      
\usepackage{natbib}        

\usepackage{amsmath}
\usepackage{amssymb}

\usepackage{accents}

\usepackage{mathtools}
\usepackage{bbold}

\usepackage{listings}

\usepackage{tikz}
\usetikzlibrary{arrows,positioning}
\usepackage{pgfplots}
\usepackage[]{caption} 
\usepackage{wrapfig}   

\usepackage{siunitx}

\pgfplotsset{compat=1.16}

\usepackage{amsthm}
\newtheorem{theorem}{Theorem}[section]
\newtheorem{lemma}[theorem]{Lemma}

\newtheorem{definition}[theorem]{Definition}

\newtheorem{assumption}[theorem]{Assumption}

\begin{document}
\begin{frontmatter}

\title{Probabilistic Reachability and Invariance Computation of Stochastic Systems using Linear Programming} 

\thanks[footnoteinfo]{Work supported by the European Research Council under the Horizon 2020 Advanced under Grant 787845 (OCAL).}

\author[First]{Niklas Schmid} 
\author[First]{John Lygeros} 

\address[First]{Automatic Control Laboratory, ETH Zürich, 8092, Zürich, Switzerland (e-mail: \{nikschmid, jlygeros\}@ethz.ch).}

\begin{abstract}                
We consider the safety evaluation of discrete time, stochastic systems over a finite horizon. Therefore, we discuss and link probabilistic invariance with reachability as well as reach-avoid problems. We show how to efficiently compute these quantities using dynamic and linear programming. 
\end{abstract}


\end{frontmatter}

\section{Introduction}
The need for safety guarantees when controlling systems arises in many important fields of application, e.g., in air-traffic management \citep{application_aircraftConflict} 
and flight control \citep{application_spacecraft}. In these applications, safety is commonly defined by a certain set of safe states, in the sense that any state trajectory leaving this set during the mission is considered unsafe. Obviously, whether safety can be guaranteed depends on the initial state that the system is in during the beginning of the mission. A natural question arising is how to evaluate which set of initial states allows for a safe evolution of the state trajectory of a controlled system. For stochastic systems, safety might only be guaranteed up to some probability. In such cases the following two questions might arise: What is the maximum / minimum achievable probability that the state trajectory remains within a certain safe set (Invariance)? What is the maximum / minimum achievable probability that the state trajectory reaches a certain target set (Reachability)? Both questions can be connected by interpreting a trajectory to be safe whenever it passes through a target set at a desired point in time. Furthermore, one can combine the two notions, leading to reach-avoid problems, where the goal is to reach a target set while staying safe. 



The computation of safety for stochastic and deterministic, continuous and discrete time systems has been extensively studied in the literature (e.g., by \cite{intermediate_DpReachability}, \cite{ct_example_mitchell}, \cite{Liao}, \cite{ct_example_peyman}). This paper specifically considers nonlinear, stochastic systems, where closed form solutions for safety guarantees typically do not exist. In such cases, despite their computational complexity, dynamic programming (DP) methods have proven useful, see \cite{intermediate_DpReachability}. 

Since these dynamic programming solutions are infinite dimensional, approximation techniques using gridding, semidefinite programming \citep{similarMethods_SemidefProg} and Lagrangian methods \citep{similarMethods_Lagrangian} are applied. In the case of reach-avoid problems, \cite{similarSolutions_LpReachAviod} propose linear programming formulations, which allow for the utilization of basis function approximations. \cite{similarSolutions_GaoDiscrete} extend these results to maximum invariance problems, but assume discrete states and actions. In this paper, we generalize these linear programming formulations to minimum and maximum reachability and invariance problems with continuous states and actions. Our contributions are as follows.
\begin{itemize}
    \item First, we establish a link between probabilistic reachability and invariance. This can be considered an extension to the observations by \cite{intermediate_JohnReachInvVia} and \cite{Liao} for deterministic systems.
    \item We then utilize simplified variants of established DP recursions \citep{intermediate_ComputationalApproaches,similarSolutions_stochasticGames} to formulate linear programs (LP) for the computation of reachable and invariant sets. Our results allow for a clear perspective and intuition on probabilistic reachability and invariance and their connection.
\end{itemize}
Other notable discussions on infinite-dimensional linear programming for the invariant set computation are given by \cite{ct_example_korda} and \cite{similarMethods_PeakEstimation}, although they follow a slightly different approach.

\subsection{Notation}
We denote by $\mathbb{1}_A(x)$ the indicator function in the set $A$, where, if $x\in A$, then $\mathbb{1}_A(x)=1$ and if $x\notin A$, then $\mathbb{1}_A(x)=0$. For two sets $X,Y$ we denote by $X\setminus Y = \{x\in X: x\notin Y\}$. $P$ denotes probability, $\mathbf{E}$ expectation and $\mathcal{B}(X)$ the Borel $\sigma$-Algebra on a topological space $X$. 

\subsection{Structure}
Section \ref{sec_safetyAndStochasticSystems} acts as an introduction. We define what we mean by stochastic systems in \ref{section_stochastic_systems}, safety in \ref{section_safety}, invariance and reachability in \ref{section_invReachVia}. Section \ref{sec_computation} discusses how to compute these quantities. We first state DP recursions in \ref{section_dp}, then reformulate these recursions as an infinite dimensional linear program in \ref{section_lp}. In order to solve the infinite dimensional linear program, we discuss approximation methods in \ref{section_approx} and apply them in numerical examples in \ref{section_numerics}. We summarize our results in \ref{section_concl}.


\section{Safety of stochastic Systems}
\label{sec_safetyAndStochasticSystems}
\subsection{Stochastic Systems and Policies}
\label{section_stochastic_systems}
A discrete time stochastic system is described by the state space $\mathcal{X}\subseteq \mathbb{R}^n$, a compact Borel set $\mathcal{U}\subseteq \mathbb{R}^m$ denoting the action space, and a Borel-measurable stochastic kernel $T:\mathcal{X}\times \mathcal{X}\times \mathcal{U}\rightarrow [0,1]$, which, given $x\in\mathcal{X}, u\in\mathcal{U}$, assigns a probability measure $T(\cdot|x,u)$ on the set $\mathcal{B}(\mathcal{X})$. 

We will denote a state at time-step $k\in\mathbb{N}$ as $x_k$ and a sequence of states $x_k,\dots,x_N$ as $x_{k:N}$. The state evolves probabilistically according to the transition kernel: For $x_k\in\mathcal{X},u_k\in\mathcal{U}$ the state transitions to a state $x_{k+1}$ in the set $B\subseteq\mathcal{B}(\mathcal{X})$ with probability $T(B|x_k,u_k)$. 

A Markov policy $\pi$ is a sequence $\pi = (\mu_0, \mu_1,\dots)$ of universally measurable maps $\mu_k: \mathcal{X}\rightarrow\mathcal{U}, k=0,1,\dots$. Starting from an initial state $x_0\in\mathcal{X}$ and under the policy $\pi$ the state evolves as $x_{k+1}\sim T(\cdot|x_k,\mu_k(x_k))$.
We denote the set of Markov policies by $\Pi$.

The transition kernel $T$, initial state $x_0\in\mathcal{X}$ and policy $\pi \in \Pi$ define a unique probability measure over $\mathcal{B}(\mathcal{X}^{N+1})$ for the trajectories (see, for example, \cite{intermediate_DpReachability}).

\subsection{Safety of a Policy}
\label{section_safety}
Given an initial state $x_0\in\mathcal{X}, \pi\in\Pi$ and a Borel set $A\subseteq \mathcal{B}(\mathcal{X})$, which we will call the safe set, we wish to compute the probability $p^{\pi}_{x_0} = P(x_{0:N}\in A|\pi,x_0)$
that the system trajectory remains within $A$ for $k=0,\dots,N$. By abuse of notation, we interpret $x_{0:N}\in A$ to mean $x_k\in A$ for $k=0,\dots N$. As shown by \cite{intermediate_DpReachability}, the safety of a trajectory can be encoded as 
\begin{equation*}
    \prod_{k=0}^N \mathbb{1}_A(x_k)= \begin{cases}
        1 &\text{if $x_{0:N}\in A$,}
        \\ 0 &\text{otherwise}.
        \end{cases}
\end{equation*}
Then, for a stochastic evolution of the state trajectory the probability of safety is defined by
\begin{align*}
    p^{\pi}_{x_0} &=\int_{\mathcal{X}^N} \prod_{k=0}^N \mathbb{1}_A(x_k)P(dx_1,\dots,dx_N|x_0,\pi) \\
    & = \underset{x_1,\dots,x_N}{\mathbf{E}}\left[\prod_{k=0}^N \mathbb{1}_A(x_k)\middle|x_0,\pi\right],
\end{align*}
which can be computed through a dynamic programming recursion. Therefore, we define functions $V_k^{\pi}:\mathcal{X}\rightarrow[0,1]$ to denote $V_k^{\pi}(x_k)=P(x_{k:N}\in A|x_k,\pi)$, then at time step $N$ we trivially obtain $V_N^{\pi}(x_N)=\mathbb{1}_A(x_N)$. Moreover,
\begin{align*}
    &P(x_{k+1:N}\in A|x_{k}, \pi)\\
    &\qquad = \int_{x_{k+1}\in\mathcal{X}}\hspace{-2.5em}P(x_{k+1:N}\in A|x_{k+1}, x_{k}, \pi)P(dx_{k+1}|x_{k}, \pi)\\
    &\qquad = \int_{x_{k+1}\in\mathcal{X}}\hspace{-2.5em}P(x_{k+1:N}\in A|x_{k+1}, \pi)P(dx_{k+1}|x_{k}, \pi) \\
    & \qquad = \int_{x_{k+1}\in\mathcal{X}} \hspace{-2.5em}V_{k+1}^{\pi}(x_{k+1})T(dx_{k+1}|x_{k},\mu_{k}(x_{k})),
\end{align*}
where the second equality follows from the Markov property of the system and the last equality follows from the definitions of $V_{k}^{\pi}$ and $T$.
Consequently, 
\begin{align*}
    V_{k}^{\pi}(x_{k}) &=P(x_{k:N}\in A|x_{k}, \pi)   \\
    &= P(x_{k}\in A|x_{k},\pi)P(x_{k+1:N}\in A|x_{k}, \pi) \\
    &= \mathbb{1}\!_A(x_{k})\!\int_{x_{k+1}\in\mathcal{X}}\hspace{-2.5em}V_{k+1}^{\pi}(x_{k+1})T(dx_{k+1}|x_{k},\mu_{k}(x_{k})).
\end{align*}
Given a policy $\pi$ we can thus compute $V_{0:N}^{\pi}(x)$ recursively backwards in time via
\begin{align*}
    V_N^{\pi}(x_N)\!&=\!\mathbb{1}_A(x_N), \\
    V_{k}^{\pi}(x_{k})\!&=\!\mathbb{1}_A(x_{k})\!\int_{x_{k+1}\in\mathcal{X}}\hspace{-2.5em}\!V_{k+1}^{\pi}(x_{k+1})T(dx_{k+1}|x_{k},\mu_{k}(x_{k})),
\end{align*}
which finally yields the desired probability $p_{x_0}^{\pi}=P(x_{0:N}\in A|x_0,\pi)=V_0^{\pi}(x_0)$.

\subsection{Invariance and Reachability}
\label{section_invReachVia}
We now define the notion of invariance (probability that the trajectory remains in the set $A$) and reachability (probability that the trajectory enters the target set $A$ at least once). 

\begin{definition}
The probabilistic maximum invariant set [$I^\uparrow$], minimum invariant set [$I_\downarrow$], maximum reachable set [$R^\uparrow$] and minimum reachable set [$R_\downarrow$], are defined as
\begin{align*}
\Omega^A_{I^\uparrow,k}(p)\!&=\!\{x_k\!\in\!\mathcal{X}| \exists \pi\!\in\!\Pi\!:\!P(x_{k:N}\!\in\!A|x_k,\pi)\!\geq\!p\},   \\
\Omega^A_{I_\downarrow,k}(p)\!&=\!\{x_k\!\in\!\mathcal{X}|\forall \pi\!\in\!\Pi\!:\! P(x_{k:N}\!\in\!A|x_k,\pi)\!\geq\!p\}, \\
\Omega^A_{R^\uparrow,k}(p)\!&=\!\{x_k\!\in\!\mathcal{X}| \exists \pi\!\in\!\Pi\!:\!P(\exists x_i\!\in\!x_{k:N}\!:\!x_i\!\in\!A|x_k,\!\pi)\!\geq\!p\}, \\
\Omega^A_{R_\downarrow,k}(p)\!&=\!\{x_k\!\in\!\mathcal{X}|\forall \pi\!\in\!\Pi\!: \!P(\exists x_i\!\in\!x_{k:N}\!:\!x_i\!\in\!A|x_k,\!\pi)\!\geq\!p\}. 
\end{align*}
\end{definition}
This definition is in correspondence to the deterministic quantities by \cite{Liao}. Maximum invariant sets are often referred to as viability sets or controlled invariant sets. Minimum invariant and maximum reachable sets are often simply referred to as invariant and reachable sets.

Assuming that the maximum and minimum are attained - we will state sufficient conditions later - it holds that
\begin{align}
    \label{eq_existsForallAndSupinf1}
    \exists \pi\in\Pi: P(\cdot|\cdot,\pi)&\geq p 
    \Leftrightarrow \max_{\pi\in\Pi} P(\cdot|\cdot,\pi)\geq p, \\
    \label{eq_existsForallAndSupinf2}
    \forall \pi\in\Pi: P(\cdot|\cdot,\pi)&\geq p 
    \Leftrightarrow \min_{\pi\in\Pi} P(\cdot|\cdot,\pi)\geq p,
\end{align}
and the respective sets correspond to level sets of the functions $V^*_{I^\uparrow,k,A}(x_k) = \max_{\pi\in\Pi} P(x_{k:N}\in A|x_k,\pi)$, $V^*_{R^\uparrow,k,A}(x_k)=\max_{\pi\in\Pi} P(\exists x_i \in x_{k:N}: x_i \in A|x_k,\pi)$, $V^*_{I_\downarrow,k,A}(x_k)=\min_{\pi\in\Pi} P(x_{k:N}\in A|x_k,\pi)$ and $V^*_{R_\downarrow,k,A}(x_k)\allowbreak=\min_{\pi\in\Pi} P(\exists x_i \in x_{k:N}: x_i \in A|x_k,\pi)$.

In the following sections we will discuss how to compute these functions using linear and dynamic programming. We only consider invariance problems, since reachability problems turn out to be duals to the respective invariance problems, as shown in figure \ref{img_linksBetweenInvarianceReachability}, and formalized in the following statement.



\begin{theorem} 
\label{thm_dual_inv_reach}
\noindent Let $A^c$ denote the complement of $A$. Then
\begin{align*}
V^*_{R^\uparrow,k,A}(x_k) &= 1 - V^*_{I_\downarrow,k,A^c}(x_k), \\
V^*_{R_\downarrow,k,A}(x_k) &= 1 - V^*_{I^\uparrow,k,A^c}(x_k)
\end{align*}
\end{theorem}
\begin{proof}
We start by showing the first relation.
Note that $x_{k:N}\in A \Leftrightarrow \nexists x_i \in x_{k:N} : x_i\notin A$. Thus,
\begin{align*}
    V^*_{R^\uparrow,k,A}(x_k) &= \max_{\pi\in\Pi}P(\exists x_i \in x_{k:N}: x_i\in A|x_{k}, \pi) \\
    &= \max_{\pi\in\Pi} 1 - P(\nexists x_i \in x_{k:N}: x_i\in A|x_{k}, \pi) \\
    &= 1 - \min_{\pi\in\Pi}P(\nexists x_i \in x_{k:N}: x_i\in A|x_{k}, \pi) \\
    &= 1 - \min_{\pi\in\Pi}P(\nexists x_i \in x_{k:N}: x_i\notin A^c|x_{k}, \pi) \\
    &= 1 - \min_{\pi\in\Pi}P(x_{k:N}\in A^c|x_{k}, \pi) \\
    &= 1 - V^*_{I_\downarrow,k,A^c}(x_k)
\end{align*}
which yields the desired equality and allows to compute
\begin{align*}
    \Omega^A_{R^\uparrow,k}(p) 
    &= \{x_k\in\mathcal{X}|V^*_{R^\uparrow,k,A}(x_k)\geq p\},\\
    &= \{x_k\in\mathcal{X}|1-V^*_{I_\downarrow,k,A^c}(x_k)\geq p\}.
\end{align*}
The proof for the second relation is similar and has already been shown in \cite{intermediate_DpReachability}
\end{proof}
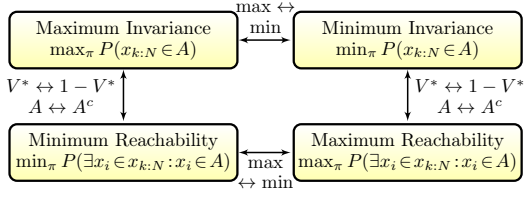
\begin{figure}[!tbp]
\centering
\resizebox{0.8\columnwidth}{!}{%
\begin{tikzpicture}[node distance=1cm, auto] 
\tikzset{
    mynode/.style={rectangle,rounded corners,draw=black, top color=white, bottom color=yellow!30,very thick, inner sep=1em, minimum size=3em,minimum width=12em, text centered},
    myarrow/.style={->, >=latex', shorten >=1pt, thick},
    myarrow2/.style={<->, >=latex', shorten >=1pt, thick},
    mylabel/.style={text width=30em, text centered} 
}  
\node[mynode,label={[align=center]center:Maximum Invariance\\ $\max_{\pi} P(x_{k:N}\!\in\!A)$}] (viability) {};  
\node[mynode, below= of viability,label={[align=center]center:Minimum Reachability\\ $\min_{\pi} P(\exists x_i\!\in\!x_{k:N}\!:\!x_i\!\in\!A)$}] (maxcost) {}; 
\node[mynode, right=1cm of viability,label={[align=center]center:Minimum Invariance\\ $\min_{\pi} P(x_{k:N}\!\in\!A)$}] (Invariance) {}; 
\node[mynode, below=of Invariance,label={[align=center]center:Maximum Reachability\\ $\max_{\pi} P(\exists x_i\!\in\!x_{k:N}\!:\!x_i\!\in\!A)$}] (Reachability) {};  
\draw[myarrow2] (viability.south) -- (maxcost.north)node[midway,right,label={[align=center]left:$V^*\leftrightarrow 1-V^*$\\$A\leftrightarrow A^c$}]{};
\draw[myarrow2] (viability.east) --  (Invariance.west)node[midway,below,label={[align=center]above:$\max\leftrightarrow$\\$\min$}] {};
\draw[myarrow2] (Invariance.south) --  (Reachability.north)node[midway,left,label={[align=center]right:$V^*\leftrightarrow 1-V^*$\\$A\leftrightarrow A^c$}]{};
\draw[myarrow2] (maxcost.east) -- (Reachability.west)node[midway,above,label={[align=center]below:$\max$\\$\leftrightarrow\min$}] {};
\end{tikzpicture} 
}
\caption{Links between invariance and reachability.} 
\label{img_linksBetweenInvarianceReachability}
\end{figure}
From now the dependency on the set $A$ will be omitted wherever possible to simplify the notation.

\section{Safety computation}
\label{sec_computation}
We now show how to solve for maximum/minimum invariance using DP recursions. Following \cite{intermediate_DpReachability} we introduce an assumption to guarantee existence of the solutions $V_{I^\uparrow,k}^*(x_k)$ and $V_{I_\downarrow,k}^*(x_k)$.
\begin{assumption}
    The set \begin{align*}
        \left\{u_{k}\!\in\!\mathcal{U}\!:\!\int_{x_{k+1}\in\mathcal{X}}\hspace{-2.5em}V_{I^\uparrow,k+1}^*(x_{k+1})T(dx_{k+1}|x_{k},u_{k})\geq \lambda\right\}
    \end{align*} and the set \begin{align*}
        \!\left\{u_{k}\!\in\!\mathcal{U}\!:\!\int_{x_{k+1}\in\mathcal{X}}\hspace{-2.5em}V_{I_\downarrow,k+1}^*(x_{k+1})T(dx_{k+1}|x_{k},u_{k})\leq \lambda\right\}
    \end{align*}
    are compact for all $x_k\in\mathcal{X}, \lambda\in\mathbb{R},k\in[0,N-1]$.
    \label{ass_attainable}
\end{assumption}
If Assumption \ref{ass_attainable} holds, the maximum in equation \eqref{eq_existsForallAndSupinf1} and minimum in equation \eqref{eq_existsForallAndSupinf2} are attained (Lemma 3.1 in \cite{intermediate_bertsekas1996stochastic}). A more intuitive, but also more restrictive condition is the continuity of $T(\cdot|\cdot,u_{k})$ with respect to $u_{k}$ (see \cite{similarSolutions_LpReachAviod}).

\subsection{Computation via Dynamic Programming}
\label{section_dp}
\begin{theorem}
    \label{thm_dp_recursion_via}
    Under Assumption \ref{ass_attainable}, the maximum invariance problem can be solved by the DP recursion (\cite{intermediate_DpReachability}) 
    \begin{align*}
            V_{I^\uparrow,N}^*(x_N)&\!=\!\mathbb{1}_A(x_N), \\
            V_{I^\uparrow,k}^*(x_{k})&\!=\!
            \max_{u_{k}}\mathbb{1}_A(x_{k})\!\int_{x_{k+1}\in\mathcal{X}}\hspace{-2.2em}V_{I^\uparrow,{k+1}}^*(x_{k+1})T(dx_{k+1}|x_{k},\!u_{k}).
    \end{align*}
    The minimum invariance problem can be solved by the DP recursion 
    \begin{align*}
            V_{I_\downarrow,N}^{*}(x_N)\!&=\!\mathbb{1}_A(x_N), \\
            V_{I_\downarrow,k}^{*}(x_{k})\!&=\!
             \min_{u_{k}}\mathbb{1}_A(x_{k})\!\int_{x_{k+1}\in\mathcal{X}}\hspace{-2.5em}V_{I_\downarrow,{k+1}}^*(x_{k+1})T(dx_{k+1}|x_{k},\!u_{k}).
    \end{align*}
\end{theorem}
\begin{proof}
    The first statement is shown in \cite{intermediate_DpReachability}.  The proof for the second statement is similar.
\end{proof}


For the maximum invariance problem, we can ease the computation by the following observation from \cite{intermediate_ComputationalApproaches}:
For all $x_{k} \in A^c,\pi\in\Pi, V_{I^\uparrow,k}^{\pi}(x_{k})=P(x_{k:N}\in A|x_{k},\pi)=0$. Consequently, we only need to compute $V_{I^\uparrow,k}^{*}(x_{k})$ for every $x_{k}\in A$, for which
\begin{align*}
    V_{I^\uparrow,k}^*(x_{k})\!&=\!\max_{u_{k}} \mathbb{1}_A(x_{k})\!\int_{x_{k+1}\in\mathcal{X}}\hspace{-2.5em}V_{I^\uparrow,{k+1}}^*(x_{k+1})T(dx_{k+1}|x_{k},\!u_{k})\\
    &=\!\max_{u_{k}} \int_{x_{k+1}\in A}\hspace{-2.5em}V_{I^\uparrow,k+1}^*(x_{k+1})T(dx_{k+1}|x_{k},u_{k}),
\end{align*}
where the last equality follows since for any $x_{k+1} \in \mathcal{X}\setminus A$ we have that $V_{I^\uparrow,k+1}^{\pi}(x_{k+1}) = 0$. Thus restricting integration to $x_{k+1}\in A$ does not change the value of the integral. It is easy to see that a similar statement holds for the minimum invariance problem, which simplifies to 
\begin{align*}
    V_{I_\downarrow,k}^*(x_{k})= \min_{u_{k}} \int_{x_{k+1}\in A}\hspace{-2.5em}V_{I_\downarrow,{k+1}}^*(x_{k+1})T(dx_{k+1}|x_{k},u_{k}).
\end{align*}

\subsection{Computation via Linear Programming}
\label{section_lp}
An interesting theoretical property of the simplified recursions is that they allow us to formulate the problems as infinite-dimensional LPs. Under assumption \ref{ass_attainable}, we relax seeking the maximum/minimum of the objective function into seeking a mimum/maximum value bounded by the set of possible objective values. For instance we can solve for the maximum invariance probability at time $k$ as
\begin{align*}
    &\min_{V_{k}} &&\int_{x_{k}\in A} \hspace{-1.6em}V_{k}(x_k)c(dx_k)\\
     & \ \text{s.t.} &&V_{k}(x_{k}) \geq \int_{x_{k+1}\in A}\hspace{-2.5em}V_{I^\uparrow,k+1}(x_{k+1})T(dx_{k+1}|x_{k},u_{k}) 
\end{align*}
where $c(\cdot)$ is a nonnegative measure that assigns positive mass to all open subsets of $\mathcal{X}$ and the constraints must hold for all $x_{k}\in A, u_{k}\in\mathcal{U}$. A similar relaxation has been proposed in \cite{similarSolutions_LpReachAviod} to formulate reach-avoid problems as LPs. Concatenating the constraints for different times $k$, we can solve invariance problems for all time steps in a single linear program, leading to 
\begin{align*}
    &\min_{V_{0:N}} &&\sum_{k=0}^{N-1}\int_{x_{k}\in A} \hspace{-1.6em}V_{k}(x_{k})c(dx_{k})\\
     & \ \text{s.t.} &&V_{k}(x_{k}) \geq \int_{x_{k+1}\in A}\hspace{-2.5em}V_{k+1}(x_{k+1})T(dx_{k+1}|x_{k},u_{k}) \\  
    &&&V_{N}(x_N) = 1
\end{align*}
for the maximum invariance problem and
\begin{align*}
    &\max_{V_{0:N}} &&\sum_{k=0}^{N-1}\int_{x_k\in A}\hspace{-1.6em} V_{k}(x_k)c(dx_k)\\
     & \ \text{s.t.} &&V_{k}(x_{k}) \leq \int_{x_{k+1}\in A}\hspace{-2.5em}V_{k+1}(x_{k+1})T(dx_{k+1}|x_{k},u_{k}) \\  
    &&& V_{N}(x_N) = 1
\end{align*}
 for the minimum invariance problem, where the constraints must hold for all $x_{k},x_{N}\in A, u_{k}\in\mathcal{U}, k\in[0,N-1]$. We denote by $V_{I^\uparrow,0:N}^{LP}(x_k)$ and $V_{I_\downarrow,0:N}^{LP}(x_k)$ the respective minimizers of the LPs.


Before we state our main results we need the following Lemma (see also Theorem 1a in \cite{similarSolutions_LpReachAviod} for reach-avoid problems). 
\begin{lemma}
\label{lemma_no_feasible_more_optimal}
    There is no feasible solution of the LP such that there exists $k\in[0,N]$ and $x_{k}\in A$ such that $V_{I^\uparrow,k}^{LP}(x_{k})<V_{I^\uparrow,k}^*(x_{k})$. %
\end{lemma}%
\begin{proof}
    Assume, for the sake of contradiction, that there exists a feasible solution to the LP, where, for some $k\in[0,N-1]$, $V_{I^\uparrow,k+1}^{LP}(\cdot)=V_{I^\uparrow,k+1}^*(\cdot)$ and there exists some $x_k$ such that $V_{I^\uparrow,k}^{LP}(x_{k})<V_{I^\uparrow,k}^*(x_{k})$. Let $\pi^*=\{\mu_{0}^*,\dots,\mu_{N-1}^*\}$ be the optimal policy obtained from the DP recursion. Since the LP solution is feasible we have for all $k\in[0,N-1], x_{k}\in A, u_{k}^*=\mu_{k}^*(x_k)$ 
    \begin{align*}
        V_{I^\uparrow,k}^{LP}(x_{k}) &\geq  \int_{x_{k+1}\in A}\hspace{-2.5em}V_{I^\uparrow,{k+1}}^{LP}(x_{k+1})T(dx_{k+1}|x_{k},u_{k}^*)\\
        &= \int_{x_{k+1}\in A}\hspace{-2.5em}V_{I^\uparrow,k+1}^*(x_{k+1})T(dx_{k+1}|x_{k},u_{k}^*) \\ 
        &= V_{I^\uparrow,k}^*(x_{k})\\
        &> V_{I^\uparrow,k}^{LP}(x_{k})
    \end{align*}
    We obtain $V_{I^\uparrow,k}^{LP}(x_{k})>V_{I^\uparrow,k}^{LP}(x_{k})$, which is a contradiction.

    Applying prior result recursively starting from the given terminal value $V_{I^\uparrow,N}^{LP}(x_N)\allowbreak=V^*_{I^\uparrow,N}(x_N)=1$ yields the claim for all $k\in[0,N]$.  
\end{proof}
In fact, the opposite relation, that $V_{I_\downarrow,k}^{LP}(x_{k})>V_{I_\downarrow,k}^*(x_{k})$ is impossible, holds for the minimum invariance case. The proof is very similar so we will skip it for brevity.

We are now ready to state our main result.
\begin{theorem}
\label{theorem_LP}
Under Assumption \ref{ass_attainable} the DP recursions for the viability and invariance problem have a solution. Then so do the respective LPs; moreover the solutions coincide up to a set of c-measure zero. 
\end{theorem}
\begin{proof}
    For the sake of brevity we only prove the LP for the maximum invariance case. The proof for the minimum invariance case is similar (see also \cite{similarSolutions_LpReachAviod} for reach-avoid problems). 
    

    First note that the DP solution is feasible for the LP, since
    \begin{align*}
        V_{I^\uparrow,k}^*(x_{k})   = \int_{x_{k+1}\in A}\hspace{-2.5em}V_{I^\uparrow,k+1}^*(x_{k+1})T(dx_{k+1}|x_{k},\mu^*_{k}(x_{k})) 
    \end{align*}  for all $k\in[0,N-1]$ and $V^*_{I^\uparrow,N}(x_N)=1$ for all $x_N\in A$, which
    meets the constraints of the LP.    
    
    Recall that for all $k\in[0,N]$ and $x_k\in A$, any feasible solution yields $V_k^{LP}(x_k)\geq V_k^*(x_k) $ by Lemma \ref{lemma_no_feasible_more_optimal}. Thus, for all feasible solutions 
    \begin{align*}
        \sum_{k=0}^{N-1}\int_{x_k\in A} \hspace{-1.6em}V_{k}^{LP}(x_k)c(dx_k) \geq \sum_{k=0}^{N-1}\int_{x_k\in A} \hspace{-1.6em}V^*_{k}(x_k)c(dx_k). 
    \end{align*}
    Assume now that for some $k\in[0,N]$ there is a set $\{x_k\in A: V_k^{LP}(x_k)> V_k^*(x_k)\}$ with non-zero c-measure. Then this would consequently lead to a suboptimal objective value, as there is no $V_k^{LP}(x_k)< V_k^*(x_k)$ to compensate by Lemma \ref{lemma_no_feasible_more_optimal}. Since $V_k^*(x_k)$ is feasible and it yields the lowest possible objective value it is a globally optimal solution to the LP. 
\end{proof}

Since the LP solution might differ from the DP solution for a set of states of zero c-measure, the optimal solution to the LP does not have to be unique. Moreover, if $V^{LP}_{I^\uparrow,k}(x_k)$ and $V^{*}_{I^\uparrow,k}(x_k)$ are solutions to the LP and DP recursions, respectively, then by optimality of the DP solution
there does not exist $\pi\in\Pi,k\in[0,N],x_k\in A:V^{LP}_{I^\uparrow,k}(x_k)=V^{\pi}_{I^\uparrow,k}(x_k) >V^{*}_{I^\uparrow,k}(x_k)$. 

Next, we consider the set of optimal policies as 
\begin{equation*}
    \Pi^*\!=\!\left\{\pi\!\in\!\Pi\!:\!\sum_{k=0}^{N-1}\int_{x\in A} \hspace{-1.3em}V_{I^\uparrow,k}^{\pi}(x)c(x)\!=\!\sum_{k=0}^{N-1}\int_{x\in A} \hspace{-1.3em}V_{I^\uparrow,k}^*(x)c(x)\right\}.
\end{equation*}

One might be tempted to think that an optimal policy $\pi^*\in\Pi^*$ is given by the active constraints of the LP (if they exist). While this is true for countable state problems with a c-measure that is strictly positive on every state, it is not always true for infinite state problems. As an example, assume that $V_{0:N}^{LP}$ is an optimal solution to the LP. We denote by $\mathbb{B}_k=\{x_k\in A: V_{k}^{LP}(x_{k})>V^*_{k}(x_{k})\}$, which must have zero c-measure for the objective to be optimal. 
    Interestingly, for states $x_{k}\in A\setminus\mathbb{B}_{k}$ the optimal inputs $u_{k}^*=\mu_{k}^*(x_{k})$ corresponding to the optimal policy obtained from the DP recursion must be contained within the active constraints, since for these inputs 
    \begin{align*}
         \int_{x_{k+1}\in A}\hspace{-2.5em}V_{k+1}^{LP}&(x_{k+1})T(dx_{k+1}|x_{k},u_{k}^*)= {V}_{k}^{LP}(x_{k})\\&={V}^*_{k}(x_{k})   = \int_{x_{k+1}\in A}\hspace{-2.5em}V^*_{k+1}(x_{k+1})T(dx_{k+1}|x_{k},u_{k}^*).
    \end{align*}
    Note that these inputs must consequently yield zero $T$-measure on $\mathbb{B}_{k+1}$. However, there may also exist an input $\tilde{u}_{k}\in\mathcal{U}$ corresponding to a tight constraint with
    \begin{align*}
          \int_{x_{k+1}\in A}\hspace{-2.5em}V^*_{k+1}&(x_{k+1})T(dx_{k+1}|x_{k},\tilde{u}_{k})  < {V}^*_{k}(x_{k}) \\&=   {V}_{k}^{LP}(x_{k}) =  \int_{x_{k+1}\in A}\hspace{-2.5em}V_{k+1}^{LP}(x_{k+1})T(dx_{k+1}|x_{k},\tilde{u}_{k}).  
    \end{align*}
    This is possible since $\mathbb{B}_{k+1}$ can have non-zero $T$-measure. Consequently, choosing a policy $\pi\in\Pi$ based on active constraints may result in such suboptimal inputs $\tilde{u}_{k}=\mu_k(x_{k})$. If it does so for a non-zero c-measure set of states $x_{k}\in A$, then the policy is suboptimal, since the actual incurred invariance is 
    \begin{align*}
          {V}^{\pi}_{k}(x_{k}) &= \int_{x_{k+1}\in A}\hspace{-2.5em}V^{\pi}_{k+1}(x_{k+1})T(dx_{k+1}|x_{k},\tilde{u}_{k}) \\ &\leq \int_{x_{k+1}\in A}\hspace{-2.5em}V^*_{k+1}(x_{k+1})T(dx_{k+1}|x_{k},\tilde{u}_{k})  < {V}^*_{k}(x_{k}). 
    \end{align*}

For the sake of completeness we also provide the following results from \cite{similarSolutions_LpReachAviod} dealing with reach-avoid problems, where one tries to reach a target set $\mathcal{T}\subseteq A$ while avoiding to step out of the safe set $A$. This means that we want to maximize $P(\exists x_i\!\in\!x_{k:N}\!:\!x_i\!\in\!\mathcal{T} \ \text{and} \ x_{k:i}\in A|x_k,\pi)$. This probability is one for every state $x_k\in \mathcal{T}$ and zero for every state $x_k\notin A$. For all $x_k\in A\setminus\mathcal{T}$, the maximum reach-avoid probability is given by

\begin{align*}
    &\min_{V_{0:N}} &&\sum_{k=0}^{N-1}\int_{x\in A} \hspace{-1.3em}V_{k}(x)c(dx)\\
     &\ \text{s.t.} &&  V_{k}(x_{k}) \geq \int_{x_{k+1}\in A\setminus \mathcal{T}}\hspace{-3.6em}V_{k+1}(x_{k+1})T(dx_{k+1}|x_{k},u_{k}) \\&&& \qquad + \int_{x_{k+1}\in\mathcal{T}}\hspace{-2.5em}T(dx_{k+1}|x_{k},u_{k})  \\ 
    &&& V_{N}(x_N) = 0,
\end{align*}
where the constraints must hold for all $x_{k},x_{N}\in A\setminus \mathcal{T}, u_{k}\in\mathcal{U}, k\in[0,N-1]$.

\subsection{Approximation Techniques}
\label{section_approx}
Several techniques have been proposed to approximate the solution to the DP/LP formulations. The state and input space may be divided into a finite number of sets, each of which is represented by a discrete point; this effectively amounts to approximating the value functions in the class of piecewise constant functions and replacing integration by summation over a finite set \citep{intermediate_DpReachability}. The computation then reduces to the case of finite states and actions. Moreover, in this class Assumption \ref{ass_attainable} is always fulfilled and the DP and LP solutions always coincide. 

More generally, the approach proposed in \cite{intermediate_ADPforReach} approximates the value functions in a subspace spanned by a finite number of basis functions, for example Radial Basis Functions (RBF). The optimization variables of the LP then become parameters of the basis functions. By sampling constraints one can obtain probabilistic guarantees based on the scenario approach. To reconstruct the optimal policy at a given state one has to sample inputs and then compute and compare the corresponding invariance or reachability probabilities. This can by circumvented by computing the state-action value functions $Q_k(x_k,u_k)$ in addition to $V_k(x_k)$. The linear program for the maximum invariance computation is then solved stagewise for all $k\in[0,N-1]$ and reads as
\begin{align*}
    &\min_{Q_{k}} &&\int_{x_{k}\in A,u_k\in \mathcal{U}} \hspace{-3.8em}Q_{k}(x_{k},u_k)c(dx_{k},du_k)\\
     &\ \text{s.t.} && Q_{k}(x_{k},u_k) \geq \int_{x_{k+1}\in A}\hspace{-2.5em}V_{k+1}(x_{k+1})T(dx_{k+1}|x_{k},u_{k}), 
\end{align*}
where the constraints must hold for all $x_k\in A$ and $u_k\in \mathcal{U}$, and we denote by $V_{k}(x_{k})=\min_{u_k\in\mathcal{U}}Q_{k}(x_{k},u_k)$, where $V_{N}(x_{N})=1$ for all $x_N \in A$. The optimal policy can then be recovered as the minimizing $u_k$ for the respective $Q_k$ at a given state $x_k$. Again, there may be a zero c-measure set of states and inputs for which the Q-function does not attain its feasible minimum.

An inherent problem of these approximations is that the computational complexity of the LP grows exponentially with the state and input space dimensionality.

\section{Numerical example}
\label{section_numerics}
As a numerical example we consider an autonomous robot that cannot rotate while moving. Thus, the robot first rotates to a desired orientation $u_k$ at its current location $x_k=(x_k^x,x_k^y)^{\top}$ and afterwards moves for a certain distance, which we fix at $\SI{3}{\meter}$ for simplicity. However, due to measurement noise, the true orientation will be $\theta_k = u_k + w_k$, with $w_k\sim \mathcal{N}(0,\pi/5)$, leading to
\begin{align*}
    \begin{bmatrix}
        x_{k+1}^x\\
        x_{k+1}^y
    \end{bmatrix} &= \begin{bmatrix}
        x_k^x + 3sin(\theta_k)\\
        x_k^y + 3cos(\theta_k)
    \end{bmatrix}.
\end{align*}

\begin{wrapfigure}{r}{0.4\columnwidth}
\vspace{-.8em}
\centering
\begin{tikzpicture}
\node[inner sep=0pt] (whitehead) at (0,0)
{\includegraphics[width=.35\columnwidth,clip, trim=4.7cm 3cm 6.5cm 2cm]{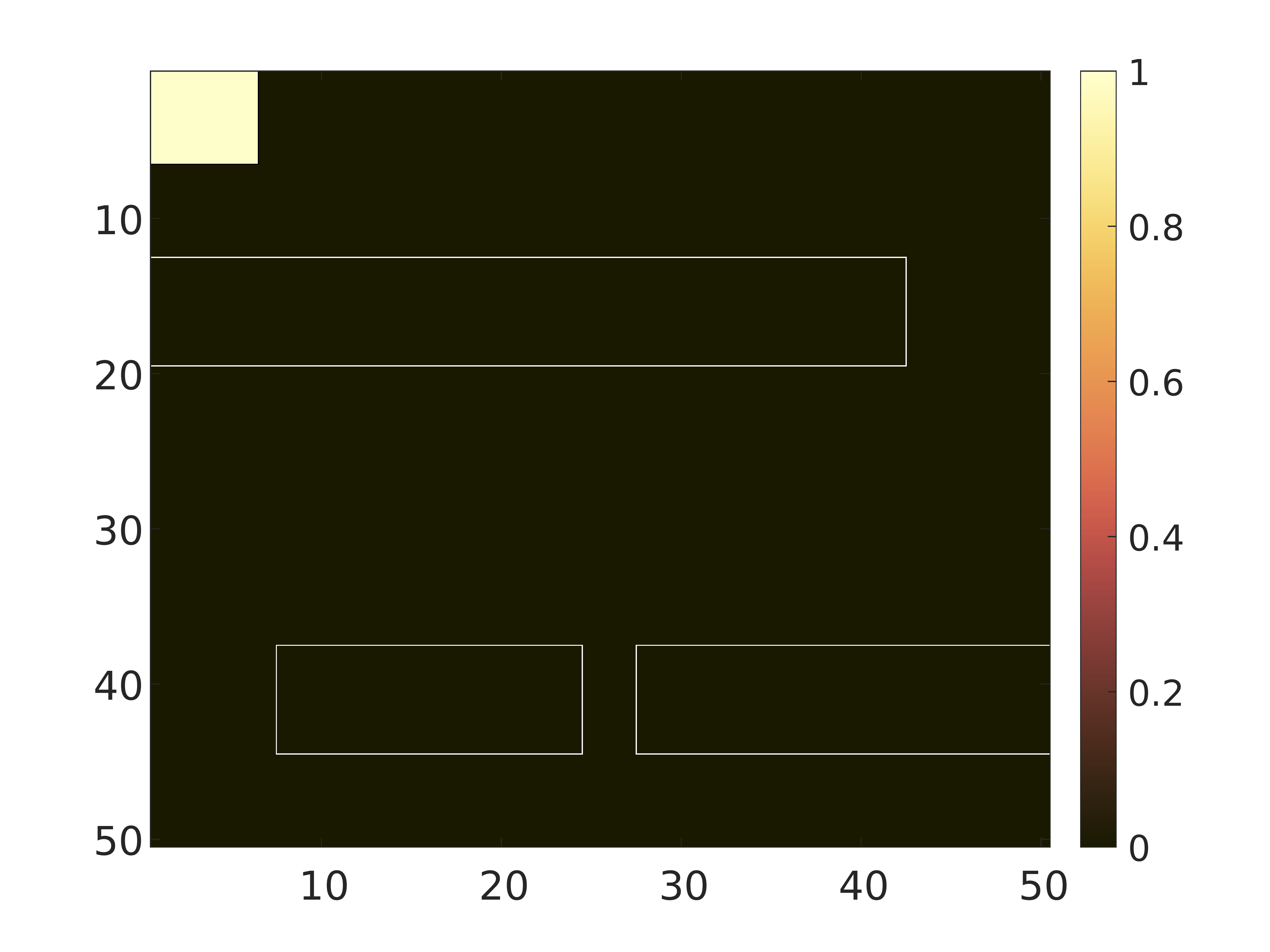}};
  \node[] at (-0.3,0.5) {\small \textcolor{white}{unsafe}};
  \node[] at (-0.61,-0.8) {\small \textcolor{white}{unsafe}};
  \node[] at (0.8,-0.8) {\small \textcolor{white}{unsafe}};
  \node[] at (-0.68,1.13) {\small \textcolor{yellow}{target}};
\end{tikzpicture}
\captionsetup{format=plain,%
    labelsep=period}
\caption{The target set is marked yellow, edges of unsafe sets white.}
\label{fig_simulation_setting}
\vspace{-1em}
\end{wrapfigure}
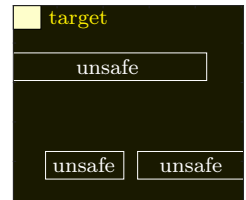

We define the state space as $\mathcal{X}=\mathbb{R}^2$, the action space as any rotation $\mathcal{U}=(0,2\pi]$ and the safe set as a room of $A=[0,50]\times[0,50]\SI{}{\meter\squared}$ excluding some additional interior walls. The robot must reach the top left corner of the room without hitting any walls (see figure \ref{fig_simulation_setting}).

We compute the reach-avoid probability of the robot for $100$ time steps. Therefore, we discretize the room into $\SI{1}{\meter}\times\SI{1}{\meter}$ blocks and the action space into $18$ actions. To compute the transition kernel, we simulate the state transition at every state action pair $1000$ times with random samples of $w_k$. In the second example we use gaussian radial basis functions $\phi_{\epsilon,c}(x)=e^{-\epsilon^2||x-c||_2^2}$. We evenly place $15\times 15$ basis functions across the state space and choose $\epsilon=\frac{15}{50}$. To generate constraints for the LP, we randomly sample $800$ states, $30$ inputs per state and $15$ samples of $w_k$ for every state action pair. For simplicity, we compute the results sequentially stage-wise instead of solving all time-steps in a single LP. The results are shown in figure \ref{fig_numericalResults}. The RBF-based solution shows artifacts visible as darker gaps and tends to be more optimistic than the gridding approach, i.e., it generally yields higher probability values. This is especially notable in the time-step $k=N-100$, where the gridding approach computes a reduced reach-avoid probability at the edges and in between the two lower walls.

\begin{figure}[!hbtp]
\centering
\resizebox{0.97\columnwidth}{!}{%
\begin{tikzpicture}
\node[inner sep=0pt] (whitehead) at (0,-13)
{\includegraphics[width=1.2\textheight,height=110mm,clip, trim=0cm 0cm 0cm 0cm]{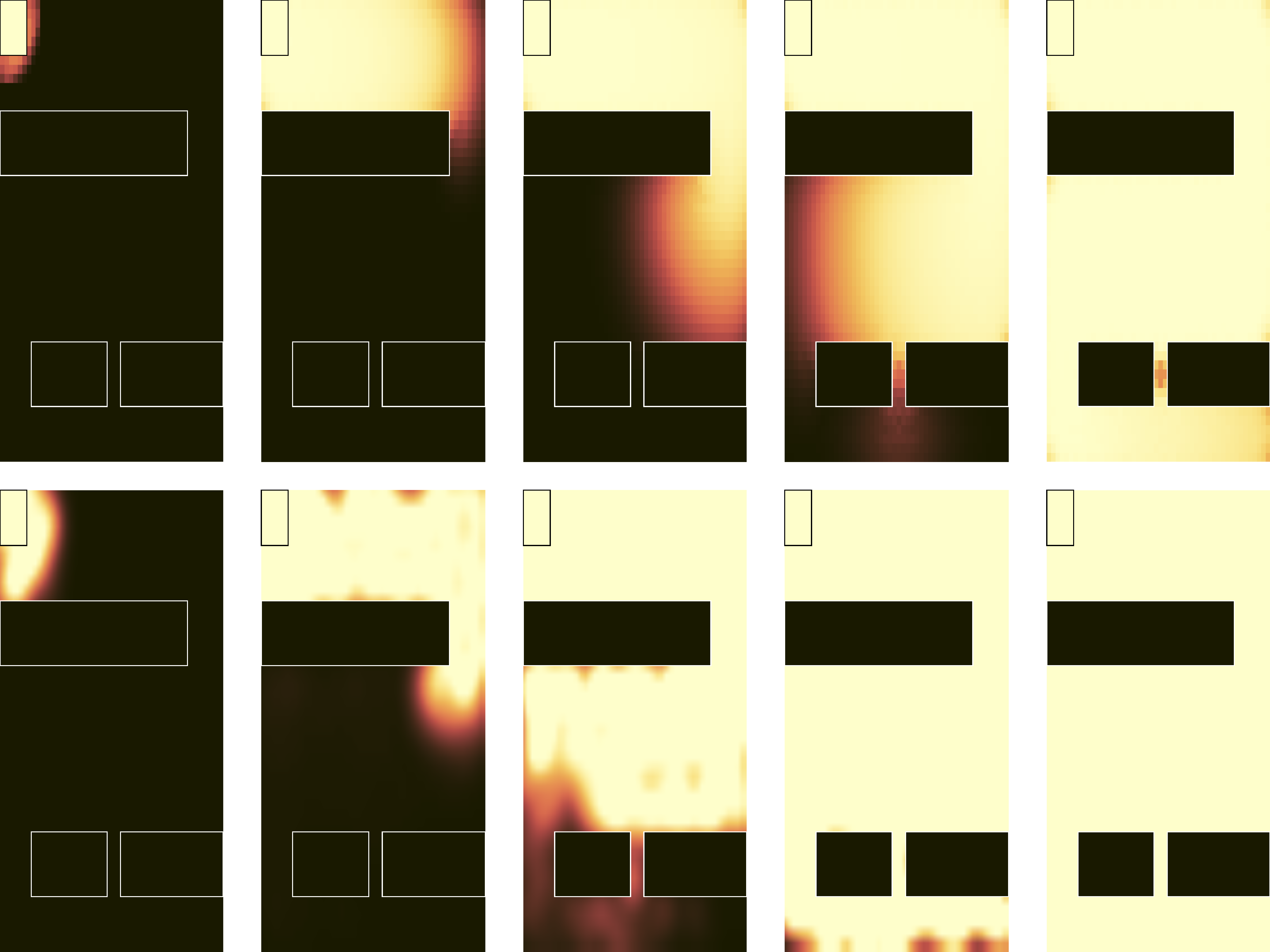}};
  \node[rotate=90] at (-16,-10) {\Huge Discretization};
  \node[rotate=90] at (-16,-16) {\Huge RBFs};
  \node[] at (-12,-19.5) {\Huge $k=N-3$};
  \node[] at (-6,-19.5) {\Huge $k=N-20$};
  \node[] at (0,-19.5) {\Huge $k=N-30$};
  \node[] at (6,-19.5) {\Huge $k=N-40$};
  \node[] at (12,-19.5) {\Huge $k=N-100$};
\node[inner sep=0pt] (russell) at (16.3,-13.2)
{\includegraphics[height=.51\textheight,clip, trim=34.5cm 1cm 2cm 1cm]{numericsSetting.pdf}};
\end{tikzpicture}
}
\caption{The plots show reach-avoid probabilities, which are computed backwards in time for $100$ time-steps using an RBF and discretization based approach. The robot aims to reach the upper left corner while avoiding to hit any walls. Initial states with high probabilities of achieving this goal until the terminal time correspond to bright areas. The further back in time (i.e., the further right), the higher these probabilities.}
\label{fig_numericalResults}
\end{figure}

To compare both approximation techniques further, the computation has been carried out with a different number of basis functions and grid densities. Grid100 to Grid25 denote results when using gridding with $100\times 100$ and $25\times 25$ representative points, respectively. RBF20 to RBF5 denote results when using $20\times 20$ to $5\times 5$ basis functions over the state space. The value of $\epsilon$ and the number of state-samples has been varied accordingly. Figure \ref{fig_errorPlot} shows the two norm distance of the respective results evaluated at $100\times 100$ grid points to those of Grid100. With increasing fineness of the grid and with increasing number of basis functions, the results are expected to converge to the true reach-avoid probabilities. Indeed, with increasing grid density and number of basis functions the plot shows a decreasing difference to the results of Grid100. Interestingly, in our simulation, the gridding based approach has been less sensitive to the approximation density than the basis function based approach. In addition, the computation times when using gridding have been significantly smaller (for $20$ time steps: $\SI{90}{\second}$ for Grid100, $\SI{4}{\second}$ for Grid25, $\SI{27}{\second}$ for RBF5, $\SI{5355}{\second}$ for RBF20 on a Surface Pro 8 i7). This is mainly due to the fact that the transition kernel is computed once for the gridding based approach and stored as a matrix at the cost of high memory usage, while the constraints in the basis function based approach have been resampled at every time-step. However, storing the full transition matrix is only feasible for low dimensional systems since its size scales exponentially with the state dimensionality. 

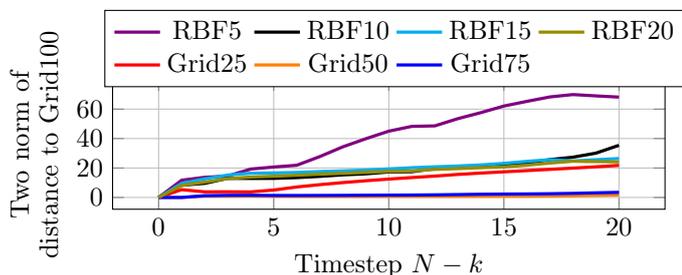
\begin{figure}[hbtp]
    \centering
\begin{tikzpicture}[]
	\begin{axis}[legend columns=2,transpose legend, height=3.3cm,width=\columnwidth,
		grid=major,legend style={at={(-0.015,1.24)},anchor=west},ylabel style={align=center}, ylabel={Two norm of \\distance to Grid100}, xlabel={Timestep $N-k$},ymax=80,ytick={0,20,40,60}
	]
	\addplot[solid,violet, very thick] coordinates {
 (0,0)
(1.000000,11.715170)(2.000000,13.678249)(3.000000,14.535770)(4.000000,19.276676)(5.000000,20.672796)(6.000000,21.864773)(7.000000,27.740496)(8.000000,34.353127)(9.000000,39.866157)(10.000000,45.011787)(11.000000,48.270004)(12.000000,48.604143)(13.000000,53.510549)(14.000000,57.597874)(15.000000,62.069782)(16.000000,65.255064)(17.000000,68.268922)(18.000000,69.866223)(19.000000,68.964666)(20.000000,68.149624)
	};
  	\addplot[solid,red, very thick] coordinates { 
 (0,0)
(1.000000,5.291503)(2.000000,3.807873)(3.000000,3.862900)(4.000000,3.776968)(5.000000,5.047502)(6.000000,7.109302)(7.000000,8.685960)(8.000000,10.094092)(9.000000,11.344255)(10.000000,12.476896)(11.000000,13.517292)(12.000000,14.508305)(13.000000,15.641857)(14.000000,16.570673)(15.000000,17.399430)(16.000000,18.229214)(17.000000,19.078672)(18.000000,19.923678)(19.000000,20.810603)(20.000000,21.786087)
	};
 	\addplot[solid,black, very thick] coordinates { 
 (0,0)
(1.000000,8.290658)(2.000000,9.563205)(3.000000,12.993734)(4.000000,12.917478)(5.000000,12.954810)(6.000000,13.463215)(7.000000,14.311604)(8.000000,15.348041)(9.000000,16.034984)(10.000000,17.318412)(11.000000,17.469503)(12.000000,19.638185)(13.000000,20.116659)(14.000000,21.776451)(15.000000,21.463684)(16.000000,23.328410)(17.000000,25.719220)(18.000000,27.332726)(19.000000,30.097341)(20.000000,35.414818)
	};
  	\addplot[solid,orange, very thick] coordinates { 
 (0,0)
(1.000000,0.000000)(2.000000,0.874929)(3.000000,0.886850)(4.000000,0.916437)(5.000000,0.816083)(6.000000,0.731366)(7.000000,0.704119)(8.000000,0.704300)(9.000000,0.706917)(10.000000,0.726811)(11.000000,0.744637)(12.000000,0.744292)(13.000000,0.760441)(14.000000,0.793461)(15.000000,0.832653)(16.000000,0.875799)(17.000000,0.952512)(18.000000,1.133541)(19.000000,1.405023)(20.000000,1.691402)
	};
  	\addplot[solid,cyan,  very thick] coordinates { 
 (0,0)
(1.000000,9.563545)(2.000000,12.658935)(3.000000,15.224938)(4.000000,16.448855)(5.000000,16.625421)(6.000000,16.976756)(7.000000,17.590385)(8.000000,18.050132)(9.000000,18.589653)(10.000000,19.369730)(11.000000,20.196311)(12.000000,20.809829)(13.000000,21.366588)(14.000000,22.106851)(15.000000,23.149805)(16.000000,24.509587)(17.000000,25.657271)(18.000000,24.789722)(19.000000,25.635300)(20.000000,26.435041)
	};
  	\addplot[solid,blue, very thick] coordinates { 
 (0,0)
(1.000000,0.000000)(2.000000,1.219549)(3.000000,1.437289)(4.000000,1.485328)(5.000000,1.347609)(6.000000,1.332999)(7.000000,1.370498)(8.000000,1.424847)(9.000000,1.508189)(10.000000,1.582113)(11.000000,1.640312)(12.000000,1.778646)(13.000000,1.963177)(14.000000,2.174366)(15.000000,2.296922)(16.000000,2.389443)(17.000000,2.536890)(18.000000,2.780313)(19.000000,3.144469)(20.000000,3.465376)
	};
  	\addplot[solid,olive, very thick] coordinates { 
 (0,0)
(1.000000,8.169972)(2.000000,10.444916)(3.000000,12.792263)(4.000000,13.869662)(5.000000,14.409106)(6.000000,15.093132)(7.000000,15.789578)(8.000000,16.364761)(9.000000,17.020740)(10.000000,17.709692)(11.000000,18.008254)(12.000000,19.048066)(13.000000,19.598331)(14.000000,20.199820)(15.000000,20.756813)(16.000000,21.929234)(17.000000,23.357009)(18.000000,24.593631)(19.000000,24.343291)(20.000000,24.170556)
	};
	\addlegendentry{RBF5}
	\addlegendentry{Grid25}
	\addlegendentry{RBF10}
	\addlegendentry{Grid50}
	\addlegendentry{RBF15}
	\addlegendentry{Grid75}
	\addlegendentry{RBF20}
	\end{axis}
\end{tikzpicture}
    \caption{We computed reach-avoid probabilities using gridding and basis function approximations and varied the gridding density/number of basis functions. The plot shows an increasing deviation over time between the computed reach-avoid probabilities to those obtained using gridding with $100\times 100$ representative points. }
    \label{fig_errorPlot}
\end{figure}




\section{Conclusion}
\label{section_concl}
We established a link between probabilistic invariance and reachability for stochastic systems and proposed infinite dimensional DP and LP formulations to solve for these quantities. Approximate formulations have been evaluated in a numerical example. We are aware that the LP formulations are not computationally more efficient than classical DP formulations. However, they may allow for different analysis and provide a new perspective on the problem.



\bibliography{root}             
                                                   







\end{document}